\documentclass[copyright,creativecommons]{eptcs}
\providecommand{\event}{TERMGRAPH 2020} 

\usepackage{amsmath}
\usepackage{amssymb}
\usepackage{amsthm}
\usepackage{dsfont}
\usepackage{mathrsfs}
\usepackage{stmaryrd}
\usepackage{tikz}

\title{Parallel Independence in Attributed Graph Rewriting}
\author{Thierry~Boy~de~la~Tour
\institute{Univ. Grenoble Alpes, CNRS, Grenoble INP, LIG, 38000 Grenoble, France}
\email{thierry.boy-de-la-tour@imag.fr}
}
\def\titlerunning{Parallel Independence in Attributed Graph Rewriting}
\def\authorrunning{T. Boy de la Tour}

\theoremstyle{plain}
\newtheorem{theorem}{Theorem}[section]
\newtheorem{lemma}[theorem]{Lemma}
\newtheorem{prop}[theorem]{Proposition}
\newtheorem{corollary}[theorem]{Corollary}

\theoremstyle{definition}
\newtheorem{definition}[theorem]{Definition}
\newtheorem{example}[theorem]{Example}
\newtheorem{remark}[theorem]{Remark}

\newcommand{\set}[1]{\{#1\}}

\newcommand{\imply}{\Rightarrow}

\newcommand{\tuple}[1]{(#1)}

\newcommand{\Zed}{\mathds{Z}}
\newcommand{\ensvide}{\varnothing}

\newcommand{\Part}[1]{\mathscr{P}(#1)}
\newcommand{\defeq}{\stackrel{\mathrm{\scriptscriptstyle def}}{=}}

\newcommand{\restrfc}[3]{#1\vert_{#2}^{#3}}
\newcommand{\restrf}[2]{\restrfc{#1}{#2}{}}
\newcommand{\invf}[1]{#1^{-1}}

\newcommand{\meetf}{\curlywedge}
\newcommand{\joinf}{\curlyvee}
\newcommand{\bigjoinf}{\curlyvee}

\newcommand{\Alg}{\mathcal{A}}

\newcommand{\labelf}[1]{\mathring{#1}}
\newcommand{\algebrf}[1]{\mathscr{A}_{#1}}
\newcommand{\attrf}[1]{\mathring{#1}}
\newcommand{\carrier}[1]{\lfloor{#1}\rfloor}

\newcommand{\nodes}[1]{\dot{#1}}
\newcommand{\arrows}[1]{\vec{#1}}
\newcommand{\leftf}[1]{\acute{#1}}
\newcommand{\rightf}[1]{\grave{#1}}
\newcommand{\subalg}{\mathrel{\lhd}}

\newcommand{\delgraph}[4]{#1\setminus[#2,#3,#4]}

\newcommand{\gcap}{\sqcap}
\newcommand{\gcup}{\sqcup}
\newcommand{\Gcup}{\bigsqcup}
\newcommand{\isog}{\mathrel{\simeq}}
\newcommand{\sepattr}{\mathrel{|}}

\newcommand{\am}{\alpha}

\newcommand{\mm}{\mu}
\newcommand{\nm}{\nu}

\newcommand{\liftm}[1]{\mbox{$#1\!\!\uparrow$}}

\newcommand{\Termsig}[1]{\mathscr{T}(\varSigma,#1)}
\newcommand{\Vars}{\mathscr{V}}
\newcommand{\Var}[1]{\mathrm{Var}(#1)}

\newcommand{\Matches}[2]{\mathscr{M}(#1,#2)}
\newcommand{\R}{\mathcal{R}}

\newcommand{\rulem}[1]{\mathrm{r}_{#1}}
\newcommand{\Lg}[1]{\mathrm{L}_{#1}}
\newcommand{\Kg}[1]{\mathrm{K}_{#1}}
\newcommand{\Rg}[1]{\mathrm{R}_{#1}}
\newcommand{\nodesLg}[1]{\nodes{\mathrm{L}}_{#1}}
\newcommand{\nodesKg}[1]{\nodes{\mathrm{K}}_{#1}}
\newcommand{\nodesRg}[1]{\nodes{\mathrm{R}}_{#1}}
\newcommand{\arrowsLg}[1]{\arrows{\mathrm{L}}_{#1}}
\newcommand{\arrowsKg}[1]{\arrows{\mathrm{K}}_{#1}}
\newcommand{\arrowsRg}[1]{\arrows{\mathrm{R}}_{#1}}

\newcommand{\labelfLg}[1]{\labelf{\mathrm{L}}_{#1}}
\newcommand{\labelfKg}[1]{\labelf{\mathrm{K}}_{#1}}
\newcommand{\labelfRg}[1]{\labelf{\mathrm{R}}_{#1}}
\newcommand{\attrfLg}[1]{\labelfLg{#1}}
\newcommand{\attrfKg}[1]{\labelfKg{#1}}
\newcommand{\attrfRg}[1]{\labelfRg{#1}}
\newcommand{\RImg}[2]{{#1}^{\uparrow}_{#2}}

\newcommand{\Ipgr}[2]{#1{\parallel}_{#2}}
\newcommand{\fullPRr}[1]{\mathrel{\Longmapsto_{#1}}}
\newcommand{\IPTrew}[1]{\mathrel{\Longrightarrow_{#1}}}
\newcommand{\Vd}[1]{\mathrm{V}_{\!#1}}
\newcommand{\Ad}[1]{\mathrm{A}_{#1}}
\newcommand{\ld}[1]{\ell_{#1}}

\newcommand{\llift}[1]{\ell^{\,\uparrow}_{#1}}
\newcommand{\seqr}[1]{\mathrel{\longrightarrow_{#1}}}
\newcommand{\seqrew}[1]{\mathrel{\longrightarrow_{#1}^{\star}}}

\begin{document}
\maketitle

\begin{abstract}
  In order to define graph transformations by the simultaneous
  application of concurrent rules, we have adopted in previous work a
  structure of attributed graphs stable by unions. We analyze the
  consequences on parallel independence, a property that characterizes
  the possibility to resort to sequential rewriting. This property
  turns out to depend not only on the left-hand side of rules, as in
  algebraic approaches to graph rewriting, but also on their
  right-hand side. It is then shown that, of three possible
  definitions of parallel rewriting, only one is convenient in the
  light of parallel independence.
\end{abstract}

\section{Introduction}\label{sec-intro}

The notion of parallel independence from \cite{Rosen75,EhrigK76} has
been studied mostly in the algebraic approaches to graph rewriting,
see \cite{HndBkCorradiniMREHL97}.  It basically consists in a
condition on concurrent transformations of an object that not only
guarantees but characterizes the possibility to apply the
transformations sequentially in any order such that all such sequences
of transformations yield the same result.

When two transformations are involved, with rules $r_1$ and $r_2$,
this takes the form of the diamond property and is known as the
\emph{Local Church-Rosser Problem} \cite{HndBkCorradiniMREHL97}; it
consists in finding a condition (called parallel independence) on
direct transformations $H_1 \xleftarrow{r_1} G \xrightarrow{r_2} H_2$
that is equivalent to the existence of direct transformations
$H_1\xrightarrow{r_2} H \xleftarrow{r_1} H_2$ with the same redexes,
hence to the existence of two equivalent sequences of transformations
$G\xrightarrow{r_1} H_1 \xrightarrow{r_2} H$ and
$G\xrightarrow{r_2} H_2 \xrightarrow{r_1} H$. It is obvious that non
overlapping redexes always entail parallel independence, the
difficulty of the problem is that the reverse does not hold and that,
depending on the rules, some amount of overlap may be allowed.  The
notion of parallel independence is also instrumental in defining
Critical Pairs (as pairs of transformations that are not parallel
independent) that are central in proving confluence of sets of
production rules \cite{LambersEO08,Costaetal16}.

This notion should therefore also be considered in algorithmic
approaches to graph rewriting. Indeed, the informal description of
parallel independence given above makes perfect sense out of the
algebraic approach; it is purely operational. Consider for instance
Python's multiple assignment $a,b:= b,a$, an elegant expression that
swaps the values of $a$ and $b$. We naturally understand this as a
parallel expression $a:=b \parallel b:= a$. If $a$ and $b$ have the
same value then the two assignments can be evaluated in sequence in
any order, yielding the same result independently of the chosen order;
they are parallel independent. If however they have distinct values,
the two sequential evaluations yield different results (and none
corresponds to the intended meaning); the two assignments are parallel
dependent.  Parallel dependence also typically occurs in cellular
automata when rules are applied to neighbor cells, because of the
overlap. Hence sequential applications of rules in an undetermined
order would result in non deterministic automata.

These examples show that there is a legitimate way of computing by
applying \emph{simultaneously} concurrent transformations that may not
be parallel independent, even though the result may not be reachable
by sequential transformations. Swapping the values of $a$ and $b$
cannot be performed by applying $a:= b$ or $b:= a$ sequentially. This
calls for a notion of \emph{parallel transformation} for defining the
result of such simultaneous applications of rules. 

One such transformation has been defined in \cite{BdlTE20c}, in an
algorithmic approach that is adopted here.  It is based on directed
graphs where vertices and arrows are equipped with \emph{sets} of
attributes and enables a definition of a \emph{union} of such graphs,
given in Section~\ref{sec-graphs}. This is a fundamental difference
with terms or termgraphs and leads to a natural definition of parallel
transformation in Section~\ref{sec-rules}.

The consequences of these definitions on parallel independence are
analyzed in Sections~\ref{sec-seqindep} and \ref{sec-parindep}. The
results of these sections are also from \cite{BdlTE20c}, we give them
here in a slightly simpler setting and without proofs, focusing on
comparisons with the algebraic approach to graph rewriting where
parallel independence has been originally formulated.

In Section~\ref{sec-EDP} we analyze the notion of parallel
rewriting. We first define a notion of \emph{regularity} that ensures the
absence of conflicts between concurrent rules. We then show that this
notion is too restricted to encompass parallel independence, and
generalize it to the \emph{effective deletion property} (also from
\cite{BdlTE20c}).

Section~\ref{sec-parcoh} is devoted to comparisons with the algebraic
notion of \emph{parallel coherence} from \cite{BdlTE20b}. It is shown
that its translation to the present framework, though more general
than regularity, is still too restricted to encompass parallel
independence. It is also shown to be the right algebraic translation
of the effective deletion property.
Concluding remarks and related works are presented in
Section~\ref{sec-conclusion}.

\section{Attributed Graphs}\label{sec-graphs}

We assume a many-sorted signature $\Sigma$ and a set $\Vars$ of
\emph{variables}, disjoint from $\Sigma$, such that every variable has
a $\Sigma$-sort. For any finite $X\subseteq \Vars$, $\Termsig{X}$
denotes the algebra of $\Sigma$-terms over $X$. For any
$\Sigma$-algebra $\Alg$, let $\carrier{\Alg}$ be the disjoint
union of the carrier sets of the $\Sigma$-sorts in $\Alg$.

An \emph{attributed graph} (or \emph{graph} for short) $G$ is a tuple
$\tuple{\nodes{G},\arrows{G},\leftf{G},\rightf{G},\algebrf{G},\attrf{G}}$
where $\nodes{G},\arrows{G}$ are sets whose elements are respectively
called \emph{vertices} and \emph{arrows}, $\leftf{G},\rightf{G}$ are
the \emph{source} and \emph{target} functions from $\arrows{G}$ to
$\nodes{G}$, $\algebrf{G}$ is a $\Sigma$-algebra and $\attrf{G}$ is an
\emph{attribution of $G$}, i.e., a function from
$\nodes{G}\cup \arrows{G}$ to $\Part{\carrier{\algebrf{G}}}$. The
elements of $\carrier{\algebrf{G}}$ are called \emph{attributes}, and
we assume that $\nodes{G}$, $\arrows{G}$ and $\carrier{\algebrf{G}}$
are pairwise disjoint.  $G$ is \emph{unlabeled} if
$\labelf{G}(x)=\ensvide$ for all $x\in \nodes{G}\cup \arrows{G}$, it
is \emph{finite} if the sets $\nodes{G}$, $\arrows{G}$ and
$\attrf{G}(x)$ are finite. The \emph{carrier} of $G$ is the set
$\carrier{G}\defeq \nodes{G}\cup \arrows{G}\cup
\carrier{\algebrf{G}}$.

A graph $H$ is a \emph{subgraph} of $G$, written $H\subalg G$, if the
\emph{underlying graph}
$\tuple{\nodes{H},\arrows{H},\leftf{H},\rightf{H}}$ of $H$ is a
subgraph of $G$'s underlying graph (in the usual sense),
$\algebrf{H} = \algebrf{G}$ and $\attrf{H}(x)\subseteq \attrf{G}(x)$
for all $x\in \nodes{H}\cup\arrows{H}$.

Graphs are better specified as pictures. Vertices and arrows will be
named and their attributes will be listed after each name, separated
from it by $\sepattr$ (which is omitted if the attribute is
$\ensvide$). Since graphs may not be connected, they will be
surrounded by a rectangle with rounded corners, as in:
\[H\ =\ 
  \raisebox{-3.2ex}{\begin{tikzpicture}
    \node[draw,rounded corners] at (0,0) 
    {\begin{tikzpicture}[scale=2.2]
        \node (x) at (0,0) {$x$};
        \node (y) at (1,0) {$y\sepattr 1$};
        \path[->,bend left] (x) edge node[fill=white,font=\footnotesize] {$f$} (y);
      \end{tikzpicture}};
  \end{tikzpicture}}\ \subalg\ 
  \raisebox{-4.7ex}{\begin{tikzpicture}
    \node[draw,rounded corners] at (0,0) 
    {\begin{tikzpicture}[scale=2.2]
        \node (x) at (0,0) {$x\sepattr 1$};
        \node (y) at (1,0) {$y\sepattr 0,1$};
        \node (z) at (1.5,0) {$z$};
        \path[->,bend left] (x) edge node[fill=white,font=\footnotesize] {$f$} (y);
        \path[->,bend right] (x) edge node[fill=white,font=\footnotesize] {$g\sepattr 0$} (y);
      \end{tikzpicture}};
  \end{tikzpicture}}\ =\ G\] where $H$ is the graph such that
$\nodes{H}=\set{x,y}$, $\arrows{H}=\set{f}$, $\leftf{H}(f)=x$,
$\rightf{H}(f)=y$, $\attrf{H}(x)=\attrf{H}(f)=\ensvide$,
$\attrf{H}(y)=\set{1}$ and similarly for $G$ (the $\Sigma$-algebra
$\algebrf{H}=\algebrf{G}$ must contain at least $0$ and $1$).

A \emph{morphism} $\am$ from graph $H$ to graph $G$, written
$\am:H\rightarrow G$, is a function from
$\carrier{H}$ to $\carrier{G}$ such that the restriction of $\am$ to
$\nodes{H}\cup\arrows{H}$ is a morphism from $H$'s to $G$'s underlying
graphs (that is, $\leftf{G}\circ{\am} = {\am}\circ\leftf{H}$ and
$\rightf{G}\circ{\am} = {\am}\circ\rightf{H}$, this restriction of
$\am$ is called the \emph{underlying graph morphism of $\am$}), the
restriction of $\am$ to $\carrier{\algebrf{H}}$ is a
$\Sigma$-homomorphism from $\algebrf{H}$ to $\algebrf{G}$, denoted
$\attrf{\am}$, and
$\attrf{\am}\circ\attrf{H}(x)\subseteq \attrf{G}\circ \am(x)$ for all
$x\in \nodes{H}\cup\arrows{H}$. Note that $H\subalg G$ iff
$\carrier{H} \subseteq \carrier{G}$ and the canonical injection from
$\carrier{H}$ to $\carrier{G}$ is a morphism from $H$ to $G$.  For all
$F\subalg H$, the \emph{image} $\am(F)$ is the smallest subgraph of
$G$ w.r.t. the order $\subalg$ such that $\restrf{\am}{\carrier{F}}$
is a morphism from $F$ to $\am(F)$.

An \emph{isomorphism} is a morphism that has an inverse morphism.  We
write $H\isog G$ if there is an isomorphism from $H$ to $G$.
A morphism $\mm:H\rightarrow G$ is a \emph{matching} if the underlying
graph morphism of $\mm$ is injective. For any $F\subalg H$ it is then
easy to see that
\[\mm(F)=\tuple{\,\mm(\nodes{F}),\ \mm(\arrows{F}),\ 
  \mm\circ\leftf{F}\circ\invf{\mm},\ 
  \mm\circ\rightf{F}\circ\invf{\mm},\  \algebrf{G},\ 
  \attrf{\mm}\circ\attrf{F}\circ\invf{\mm}\ }.\]
 
Given two attributions $l$ and $l'$ of $G$ let $l\setminus l'$
(resp. $l\cap l'$, $l\cup l'$) be the attribution of $G$ that maps any
$x$ to $l(x)\setminus l'(x)$ (resp. $l(x)\cap l'(x)$,
$l(x)\cup l'(x)$). If $l$ is an attribution of a subgraph
$H\subalg G$, it is implicitly extended to the attribution of $G$ that
is identical to $l$ on $\nodes{H}\cup\arrows{H}$ and maps any other
entry to $\ensvide$.

Unions of graphs can only be formed between \emph{joinable} graphs,
i.e., graphs that have a common part.  We start with a simpler notion
of joinable functions.

\begin{definition}[joinable functions]\label{def-joinablef}
  Two functions $f:D\rightarrow C$ and $g:D'\rightarrow C'$ are
  \emph{joinable} if $f(x)=g(x)$ for all $x\in D\cap D'$. Then, the
  \emph{meet} of $f$ and $g$ is the function
  $f \meetf g: D\cap D'\rightarrow C\cap C'$ that is the restriction
  of $f$ (or $g$) to $D\cap D'$. The \emph{join} $f \joinf g$ is the
  unique function from $D\cup D'$ to $C\cup C'$ such that
  $f=\restrf{(f\joinf g)}{D}$ and $g=\restrf{(f\joinf g)}{D'}$.

  For any set $I$ and any $I$-indexed family
  $(f_i:D_i\rightarrow C_i)_{i\in I}$ of pairwise joinable functions,
  let $\bigjoinf_{i\in I}f_i$ be the only function from
  $\bigcup_{i\in I}D_i$ to $\bigcup_{i\in I}C_i$ such that
  $f_i = \restrf{\big(\bigjoinf_{i\in I}f_i\big)}{D_i}$ for all
  $i\in I$.
\end{definition}

We see that any two restrictions $\restrf{f}{A}$ and $\restrf{f}{B}$
of the same function $f$ are joinable, and then
$\restrf{f}{A}\meetf \restrf{f}{B} = \restrf{f}{A\cap B}$ and
$\restrf{f}{A}\joinf \restrf{f}{B} = \restrf{f}{A\cup B}$. Conversely,
if $f$ and $g$ are joinable then each is a restriction of $f\joinf g$.

\begin{definition}[joinable graphs]\label{def-joinableg} 
  Two graphs $H$ and $G$ are \emph{joinable} if
  $\algebrf{H}=\algebrf{G}$,
  $\nodes{H}\cap\arrows{G} = \arrows{H}\cap\nodes{G} = \ensvide$, and
  the functions $\leftf{H}$ and $\leftf{G}$ (and similarly
  $\rightf{H}$ and $\rightf{G}$) are joinable.  We can then define the
  graphs
\begin{eqnarray*}
  H\gcap G &\defeq & \tuple{\ \nodes{H}\cap\nodes{G},\
                     \arrows{H}\cap\arrows{G},\
                     \leftf{H}\meetf\leftf{G},\
                     \rightf{H}\meetf\rightf{G},\
                     \algebrf{H},\  \attrf{H}\cap\attrf{G}\ },\\ 
  H\gcup G &\defeq & \tuple{\ \nodes{H}\cup\nodes{G},\
                     \arrows{H}\cup\arrows{G},\
                     \leftf{H}\joinf\leftf{G},\
                     \rightf{H}\joinf\rightf{G},\
                     \algebrf{H},\ \attrf{H}\cup\attrf{G}\ }. 
\end{eqnarray*}
Similarly, if $(G_i)_{i\in I}$ is an $I$-indexed family of graphs that
are pairwise joinable, and $\Alg$ is an algebra such that
$\Alg=\algebrf{G_i}$ for all $i\in I$, then let
\[\displaystyle \Gcup_{i\in I}G_i \ \defeq\ \tuple{\ \bigcup_{i\in I}\nodes{G_i},\
    \bigcup_{i\in I}\arrows{G_i},\ \bigjoinf_{i\in
      I}\leftf{G_i},\ \bigjoinf_{i\in I}\rightf{G_i},\ \Alg,\ \bigcup_{i\in
      I}\attrf{G_i}\ }.\]
\end{definition}

It is easy to see that these structures are graphs: the sets of
vertices and arrows are disjoint and the adjacency functions have the
correct domains and codomains. If $I=\ensvide$ the chosen algebra
$\Alg$ is generally obvious from the
context. Note that if $H$ and $G$ are joinable
then $H\gcap G \subalg H \subalg H\gcup G$. Similarly, if the $G_i$'s are pairwise joinable then 
$G_j\subalg \Gcup_{i\in I}G_i$ for all $j\in I$.
We see that any two subgraphs of $G$ are joinable, and that
$H\subalg G$ iff $H\gcap G = H$ iff $H\gcup G = G$. These operations
are commutative and, on triples of pairwise joinable graphs, they are
associative and distributive over each other. For any two graphs $H,G$
there exists $G'\isog G$ such that $H$ and $G'$ are joinable (one
possibility is to take $\nodes{G}'\cap \arrows{H}=\ensvide$ and
$\arrows{G}'\cap (\nodes{H}\cup\arrows{H})=\ensvide$).

For any sets $V$, $A$ and attribution $l$, we say that \emph{$G$ is
  disjoint from $V,A,l$} if $\nodes{G}\cap V=\ensvide$,
$\arrows{G}\cap A = \ensvide$ and $\attrf{G}(x)\cap l(x) = \ensvide$
for all $x\in\nodes{G}\cup\arrows{G}$. We write
$\delgraph{G}{V}{A}{l}$ for the largest subgraph of $G$
(w.r.t. $\subalg$) that is disjoint from $V,A,l$. This provides a natural
way of removing objects from an attributed graph. It is easy to see
that this subgraph always exists (it is the union of all subgraphs of
$G$ disjoint from $V,A,l$), hence rewriting steps will not be
restricted by a \emph{gluing condition} as in the Double-Pushout
approach (see \cite{EhrigEPT06}).

\section{Applying Rules in Parallel}\label{sec-rules}

\begin{definition}[rules, matchings]\label{def-rules}
  For any finite $X\subseteq \Vars$, a \emph{$(\Sigma,X)$-graph} is a
  finite graph $G$ such that $\algebrf{G} = \Termsig{X}$. Let
  $\Var{G}\ \defeq\
  \bigcup_{x\in\nodes{G}\cup\arrows{G}}\big(\bigcup_{t\in\attrf{G}(x)}
  \Var{t}\big)$, where $\Var{t}$ is the set of variables occurring in
  $t$.
  
  A \emph{rule} $r$ is a triple $\tuple{L,K,R}$ of $(\Sigma,X)$-graphs
  such that $L$ and $R$ are joinable, $L\gcap R\subalg K\subalg L$ and
  $\Var{L} = X$ (see Remark~\ref{rem-rules} below). The rule $r$ is \emph{unlabeled} if $L$, $K$ and $R$ are unlabeled.

  A \emph{matching of $r$ in} a graph $G$ is a matching $\mm$ from $L$
  to $G$ that is \emph{consistent}, i.e., such that
  $\attrf{\mm}(\attrf{L}(x)\setminus \attrf{K}(x))\cap
  \attrf{\mm}(\attrf{K}(x)) = \ensvide$ (or equivalently
  $\attrf{\mm}(\attrf{L}(x)\setminus \attrf{K}(x)) =
  \attrf{\mm}(\attrf{L}(x))\setminus \attrf{\mm}(\attrf{K}(x))$) for
  all $x\in \nodes{K}\cup\arrows{K}$.  We denote $\Matches{r}{G}$ the
  set of all matchings of $r$ in $G$ (they all have domain
  $\carrier{L}$).

  We consider finite sets $\R$ of rules such that for all $r,r'\in\R$,
  if $\tuple{L,K,R} = r \neq r' = \tuple{L',K',R'}$ then
  $\carrier{L}\neq\carrier{L'}$, so that
  $\Matches{r}{G}\cap \Matches{r'}{G} = \ensvide$ for any graph $G$;
  we then write $\Matches{\R}{G}$ for
  $\biguplus_{r\in\R}\Matches{r}{G}$. For any $\mm\in \Matches{\R}{G}$
  there is a unique rule $\rulem{\mm}\in\R$ such that
  $\mm\in \Matches{\rulem{\mm}}{G}$, and its components are denoted
  $\rulem{\mm} = \tuple{\Lg{\mm}, \Kg{\mm}, \Rg{\mm}}$.
\end{definition}

\begin{remark}\label{rem-rules}
  If $X$ were allowed to contain a variable $v$ not occurring in $L$,
  then $v$ would freely match any element of $\algebrf{G}$ and the set
  $\Matches{r}{G}$ would contain as many matchings with essentially
  the same effect. Also note that $\Var{R}\subseteq\Var{L}$, $R$ and
  $K$ are joinable and $R\gcap K = L\gcap R$. The fact that $K$ is not
  required to be a subgraph of $R$ allows the possible deletion by
  other rules of data matched by $K$ but not by $R$. 
\end{remark}

A rewrite step may involve the creation of new vertices in a graph,
corresponding to the vertices of a rule that have no match in the
input graph, i.e., those in $\nodes{R}\setminus\nodes{L}$ (or
similarly may create new arrows). These vertices should really be new,
not only different from the vertices of the original graph but also
different from the vertices created by other transformations
(corresponding to other matchings in the graph).  We
simply reuse the vertices $x$ from $\nodes{R}\setminus\nodes{L}$ by
\emph{indexing} them with any relevant matching $\mm$, each time
yielding a new vertex $\tuple{x,\mm}$ which is obviously different
from any new vertex $\tuple{x,\nm}$ for any other matching
$\nm\neq\mm$, and also from any vertex of $G$. This is similar to a
construction of colimits in the category of sets.

\begin{definition}[graph $\RImg{G}{\mm}$ and matching $\liftm{\mm}$]\label{def-RImg}  
  For any rule $r=\tuple{L,K,R}$, graph $G$ and $\mm\in\Matches{r}{G}$
  we define a graph $\RImg{G}{\mm}$ together with a matching
  $\liftm{\mm}$ of $R$ in $\RImg{G}{\mm}$. We first define the sets
  \[\RImg{\nodes{G}}{\mm} \defeq {\mm}(\nodes{R}\cap
    \nodes{K})\cup ((\nodes{R} \setminus \nodes{K}) \times
    \set{\mm})\ \text{ and }\ \RImg{\arrows{G}}{\mm} \defeq
    {\mm}(\arrows{R}\cap \arrows{K})\cup ((\arrows{R} \setminus
    \arrows{K}) \times \set{\mm}).\] Next we define $\liftm{\mm}$ by:
  $\liftm{\attrf{\mm}} \defeq \attrf{\mm}$ and for all
  $x\in\nodes{R}\cup\arrows{R}$, if $x\in\nodes{K}\cup\arrows{K}$ then
  $\liftm{{\mm}}(x) \defeq {\mm}(x)$ else
  $\liftm{{\mm}}(x) \defeq \tuple{x,\mm}$. Since the restriction of
  $\liftm{{\mm}}$ to $\nodes{R}\cup\arrows{R}$ is bijective, then
  $\liftm{\mm}$ is a matching from $R$ to the graph
  \[\RImg{G}{\mm}\ \defeq\
    \tuple{\ \RImg{\nodes{G}}{\mm},\ \RImg{\arrows{G}}{\mm},\ 
      \liftm{{\mm}}\circ \leftf{R}\circ
      \invf{\liftm{{\mm}}},\ \liftm{{\mm}}\circ
      \rightf{R}\circ \invf{\liftm{{\mm}}},\ \algebrf{G},\ 
      \liftm{\attrf{\mm}}\circ \attrf{R}\circ
      \invf{\liftm{{\mm}}}\ }.\]
\end{definition}

By construction $\liftm{\mm}(R)=\RImg{G}{\mm}$, the matchings $\mm$ and
$\liftm{\mm}$ are joinable and $\mm\meetf\liftm{\mm}$ is a matching
from $R\gcap K$ to $\mm(R\gcap K)$. It is easy to see that the graph
$G$ and the graphs $\RImg{G}{\mm}$ are pairwise joinable.

For any set $M\subseteq\Matches{\R}{G}$ of matchings in a graph $G$ we
define below how to transform $G$ by applying simultaneously the rules
associated with matches in $M$. This simply consists in first removing
simultaneously all the vertices, arrows and attributes that are
matched by $\Lg{\mm}$ but not by $\Kg{\mm}$ for any $\mm\in M$, and
then in adding simultaneously all the images of the right-hand sides
$\Rg{\mm}$.

\begin{definition}[graph $\Ipgr{G}{M}$]\label{def-parallel-rew}
  For any graph $G$ and set $M\subseteq \Matches{\R}{G}$ let
  \[\Ipgr{G}{M} \defeq \delgraph{G}{\Vd{M}}{\Ad{M}}{\ld{M}} \gcup \Gcup_{\mm\in
      M}\RImg{G}{\mm}\text{ where}\]
  \[\Vd{M} \defeq \bigcup_{\mm\in M}{\mm}(\nodesLg{\mm}\setminus
    \nodesKg{\mm}),\ \Ad{M} \defeq \bigcup_{\mm\in
      M}{\mm}(\arrowsLg{\mm}\setminus \arrowsKg{\mm})\text{ and }
    \ld{M}\defeq \bigcup_{\mm\in
      M}\attrf{\mm}\circ(\attrfLg{\mm}\setminus
    \attrfKg{\mm})\circ\invf{\mm}.\]
    If $M$ is a singleton $\set{\mm}$ we write $\Ipgr{G}{\mm}$ for
    $\Ipgr{G}{M}$, $\Vd{\mm}$ for $\Vd{M}$, etc.
\end{definition}

\begin{example}\label{ex-Ipgr}
  We represent the simultaneous assignment $a,b := b,a$ by two rules
  that correspond to the simple assignments $a:=b$ and $b:=a$. For
  this we use a signature $\Sigma$ with two constants $a$ and $b$ of
  sort \texttt{identifier}, and a set of variables $\Vars$ with two
  variables $u$ and $v$ of sort \texttt{integer}. The environment is
  represented by two nodes $x$ and $y$, each attributed by an
  identifier and its value. More precisely, let $\Alg$ be the
  $\Sigma$-algebra where the sort \texttt{integer} is interpreted as
  $\Alg_{\texttt{integer}} = \Zed$, the sort \texttt{identifier} as
  $\Alg_{\texttt{identifier}} = \set{a,b}$ and each constant as
  itself. We consider the environment where $a=1$ and $b=-1$; this is
  represented by the graph
  \[G\ =\ \raisebox{-1.3ex}{\begin{tikzpicture}
        \node[draw,rounded corners] at (0,0) {$x\sepattr a,1\hspace{2em}
          y\sepattr b,-1$};
      \end{tikzpicture}}\]

  We consider the rules
  \begin{align*}
  r_1& = \tuple{\ \raisebox{-1.3ex}{\begin{tikzpicture}
  \node[draw,rounded corners] at (0,0) {$x_1\sepattr a,u\hspace{2em}
    y_1\sepattr b,v$};
  \end{tikzpicture}}\ ,\ \raisebox{-1.3ex}{\begin{tikzpicture}
  \node[draw,rounded corners] at (0,0) {$x_1\sepattr a\hspace{2em}
    y_1\sepattr b,v$};
  \end{tikzpicture}}\ ,\ \raisebox{-1.3ex}{\begin{tikzpicture}
  \node[draw,rounded corners] at (0,0) {$x_1\sepattr a,v$};
  \end{tikzpicture}}\ }\\
  r_2& = \tuple{\ \raisebox{-1.3ex}{\begin{tikzpicture}
  \node[draw,rounded corners] at (0,0) {$x_2\sepattr a,u\hspace{2em}
    y_2\sepattr b,v$};
  \end{tikzpicture}}\ ,\ \raisebox{-1.3ex}{\begin{tikzpicture}
  \node[draw,rounded corners] at (0,0) {$x_2\sepattr a,u\hspace{2em}
    y_2\sepattr b$};
  \end{tikzpicture}}\ ,\ \raisebox{-1.3ex}{\begin{tikzpicture}
  \node[draw,rounded corners] at (0,0) {$y_2\sepattr b,u$};
  \end{tikzpicture}}\ }
  \end{align*}
  that correspond to $a:=b$ and $b:=a$ respectively. We see that $r_1$
  removes the content $u$ associated to $a$ and replaces it by the
  content $v$ associated to $b$. There is exactly one matching $\mm_i$
  of rule $r_i$ in $G$ for $i=1,2$, given by
\[\begin{array}{c|cccccc}
      &x_1&y_1&a&b&u&v\\ \hline
      \mm_1&x&y&a&b&1&-1
    \end{array}\hspace{2em} 
    \begin{array}{c|cccccc}
      &x_2&y_2&a&b&u&v\\ \hline
      \mm_2&x&y&a&b&1&-1
    \end{array}\]

  Let $M=\set{\mm_1,\mm_2}$. Since no vertex or arrow is removed we
  have $\Vd{M}=\Ad{M}=\ensvide$. We also have
  \begin{align*}
    \ld{M}(x) &= \big(\attrf{\mm}_1\circ(\attrfLg{\mm_1}\setminus
                \attrfKg{\mm_1})\circ\invf{\mm_1}(x)\big) \cup \big(\attrf{\mm}_2\circ(\attrfLg{\mm_2}\setminus
                \attrfKg{\mm_2})\circ\invf{\mm_2}(x)\big)\\
              &= \attrf{\mm}_1\big(\attrfLg{\mm_1}(x_1)\setminus
                \attrfKg{\mm_1}(x_1)\big) \cup \attrf{\mm}_2\big(\attrfLg{\mm_2}(x_2)\setminus
                \attrfKg{\mm_2}(x_2)\big)\\
              &= \attrf{\mm}_1(\set{u})\cup \attrf{\mm}_2(\ensvide)\\
              &= \set{1}
  \end{align*}
  and similarly $\ld{M}(y)=\set{-1}$, so that
  $\delgraph{G}{\Vd{M}}{\Ad{M}}{\ld{M}}\ =\ \raisebox{-1.3ex}{\begin{tikzpicture}
        \node[draw,rounded corners] at (0,0) {$x\sepattr a\hspace{2em}
          y\sepattr b$};
      \end{tikzpicture}}$. Finally, we see that
    \begin{align*}
      \RImg{G}{\mm_1} &=\liftm{\mm_1}(\Rg{\mm_1}) = \liftm{\mm_1}(\ \raisebox{-1.3ex}{\begin{tikzpicture}
  \node[draw,rounded corners] at (0,0) {$x_1\sepattr a,v$};
  \end{tikzpicture}}\ ) = \raisebox{-1.3ex}{\begin{tikzpicture}
  \node[draw,rounded corners] at (0,0) {$x\sepattr a,-1$};
  \end{tikzpicture}}\\
      \RImg{G}{\mm_2} &=\liftm{\mm_2}(\Rg{\mm_2}) = \liftm{\mm_2}(\ \raisebox{-1.3ex}{\begin{tikzpicture}
  \node[draw,rounded corners] at (0,0) {$y_2\sepattr b,u$};
  \end{tikzpicture}}\ ) = \raisebox{-1.3ex}{\begin{tikzpicture}
  \node[draw,rounded corners] at (0,0) {$y\sepattr b,1$};
  \end{tikzpicture}}
    \end{align*}
and hence \[\Ipgr{G}{M}\ =\ \raisebox{-1.3ex}{\begin{tikzpicture}
        \node[draw,rounded corners] at (0,0) {$x\sepattr a\hspace{2em}
          y\sepattr b$};
      \end{tikzpicture}}\ \gcup\ \raisebox{-1.3ex}{\begin{tikzpicture}
  \node[draw,rounded corners] at (0,0) {$x\sepattr a,-1$};
  \end{tikzpicture}}\ \gcup\ \raisebox{-1.3ex}{\begin{tikzpicture}
  \node[draw,rounded corners] at (0,0) {$y\sepattr b,1$};
  \end{tikzpicture}} \ =\ \raisebox{-1.3ex}{\begin{tikzpicture}
        \node[draw,rounded corners] at (0,0) {$x\sepattr a,-1\hspace{2em}
          y\sepattr b,1$};
      \end{tikzpicture}}\] that represents the environment where
  $a=-1$ and $b=1$, i.e., where the initial values of $a$ and $b$ have
  been swapped. Note that the same transformation can obviously be
  performed by the single rule 
\[ \tuple{\ \raisebox{-1.3ex}{\begin{tikzpicture}
  \node[draw,rounded corners] at (0,0) {$x_1\sepattr a,u\hspace{2em}
    y_1\sepattr b,v$};
  \end{tikzpicture}}\ ,\ \raisebox{-1.3ex}{\begin{tikzpicture}
  \node[draw,rounded corners] at (0,0) {$x_1\sepattr a\hspace{2em}
    y_1\sepattr b$};
  \end{tikzpicture}}\ ,\ \raisebox{-1.3ex}{\begin{tikzpicture}
  \node[draw,rounded corners] at (0,0) {$x_1\sepattr a,v\hspace{2em}
    y_1\sepattr b,u$};
  \end{tikzpicture}}\ }\]
More importantly, this rule can be
  computed from $r_1$ and $r_2$ (see \cite{BdlTE20b}).
\end{example}

In Definition~\ref{def-parallel-rew} $\Ipgr{G}{M}$ is guaranteed
to be a graph since the $\gcup$ operation is only applied on joinable
graphs. Every morphism $\liftm{\mm}$ is a matching from the right-hand
side $\Rg{\mm}$ to the result $\Ipgr{G}{M}$ of the transformation.
The case where $M$ is a singleton defines the classical semantics of
one sequential rewrite step.

\begin{definition}[sequential rewriting]\label{def-seqrew}
  For any finite set of rules $\R$, we define the relation $\seqr{\R}$
  of \emph{sequential rewriting} by stating that, for all graphs $G$
  and $H$, \[G\seqr{\R}H\text{ iff there exists some }\mm\in\Matches{\R}{G}
  \text{ such that }H\isog \Ipgr{G}{\mm}.\]
\end{definition}

\section{Sequential  Independence}\label{sec-seqindep}

In the Double-Pushout approach to graph rewriting (see
\cite{EhrigEPT06}), production rules are spans
$L\leftarrow K\rightarrow R$, with two morphisms from an
\emph{interface} $K$ to the left- and right-hand sides $L$, $R$. These
objects and morphisms are taken in a category, possibly of some sort
of graphs. Direct derivations are diagrams
\begin{center}
  \begin{tikzpicture}[xscale=1.65, yscale=1.6]
    \node (G) at (-0.5,0) {$H$};
    \node (L1) at (-0.5,1) {$R$}; \node (K1) at (-1.5,1) {$K$};
    \node (R1) at (-2.5,1) {$L$};
    \node (D1) at (-1.5,0) {$D$}; \node (H1) at (-2.5,0) {$G$};
    \path[->] (K1) edge (L1);
    \path[->] (K1) edge (R1);
    \path[->] (L1) edge (G);
    \path[->] (K1) edge (D1);
    \path[->] (D1) edge (G);
    \path[->] (D1) edge (H1);
    \path[->] (R1) edge node[fill=white, font=\footnotesize] {$\mm_1$} (H1);
  \end{tikzpicture}    
\end{center}
where the two squares are \emph{pushouts}, i.e., a form of
union. Since objects, say $D$ and $R$, can always be understood modulo
isomorphisms, their union cannot be defined without specifying what
they have in common; this is the rôle of $K$ and of the morphisms from
$K$ to $D$ and $R$. If for instance $K$ is empty then the pushout $H$
is the disjoint union (or direct sum, or co-product) of $D$ and
$R$. Hence the right square adds something to $D$, and inversely the
left square removes something from $G$. Hence $H$ is obtained from $G$
by removing an image of $L$ and writing an image of $R$, with the
possibility that $L$ and $R$ share a common part given by $K$. This
very general approach has a drawback: depending on $G$ and $\mm_1$ the
object $D$ may not exist, and if it does it may not be unique.

In this approach sequential independence is a property of two
consecutive direct transformations, formulated as the existence of two
commuting morphisms $j_1$ and $j_2$ as shown below.

\begin{center}
  \begin{tikzpicture}[xscale=1.65, yscale=1.8]
    \node (L) at (0.5,1) {$L_2$}; \node (K) at (1.5,1) {$K_2$};
    \node (R) at (2.5,1) {$R_2$}; \node (G) at (0,0) {$H_1$};
    \node (D) at (1.5,0) {$D_2$}; \node (H) at (2.5,0) {$H_2$};
    \path[->] (K) edge (L);
    \path[->] (K) edge (R);
    \path[->] (L) edge node[fill=white, font=\footnotesize, near end] {$\mm_2$} (G);
    \path[->] (K) edge (D);
    \path[->] (D) edge (G);
    \path[->] (D) edge (H);
    \path[->] (R) edge (H);
    \node (L1) at (-0.5,1) {$R_1$}; \node (K1) at (-1.5,1) {$K_1$};
    \node (R1) at (-2.5,1) {$L_1$};
    \node (D1) at (-1.5,0) {$D_1$}; \node (H1) at (-2.5,0) {$G$};
    \path[->] (K1) edge (L1);
    \path[->] (K1) edge (R1);
    \path[->] (L1) edge (G);
    \path[->] (K1) edge (D1);
    \path[->] (D1) edge (G);
    \path[->] (D1) edge (H1);
    \path[->] (R1) edge node[fill=white, font=\footnotesize] {$\mm_1$} (H1);
  \path[-] (L1) edge[draw=white, line width=3pt]  (D);
  \path[-] (L) edge[draw=white, line width=3pt]  (D1);
    \path[->,dashed] (L1) edge node[fill=white, font=\footnotesize] {$j_1$} (D);
    \path[->,dashed] (L) edge node[fill=white, font=\footnotesize] {$j_2$} (D1);
  \end{tikzpicture}    
\end{center}
It is then proven by the Local Church-Rosser Theorem that the two
production rules can be applied in reverse order to $G$ and yield the
same result $H_2$ (we may call this the \emph{swapping property}). Of
course, the matchings $\mm_1:L_1\rightarrow G$ and
$\mm_2:L_2\rightarrow H_1$ are then replaced by other matchings
$\mm'_1:L_1\rightarrow H'_1$ and $\mm'_2:L_2\rightarrow G$ that are
related to $\mm_1$ and $\mm_2$. A drawback of this definition is that
it does not account for longer sequences of direct
transformations. Indeed, if three consecutive steps are given by
$\tuple{\mm_1,\mm_2,\mm_3}$, it is possible to swap $\mm_1$ with
$\mm_2$ if they are sequential independent, and similarly for $\mm_2$
and $\mm_3$, but this does not imply that $\mm_1$ and $\mm_3$ can be
swapped under these hypotheses (because the matchings, and hence the
direct transformations, are modified by the swapping operations). We
would need to express sequential independence between $\mm_1$ and
$\mm_3$, but the definition does not apply since they are not
consecutive steps. More elaborate notions of equivalence between
sequences of direct transformations are thus required (see the notion
of \emph{shift equivalence} in \cite[chapter
3.5]{HndBkCorradiniMREHL97}).

Because of the specificities of our framework (no pushouts, horizontal
morphisms are only canonical injections, and there may be no such morphism
from $K$ to $R$) we need a different definition of sequential
independence. It is natural to think of the swapping property itself
as the definition of sequential independence, since it describes the
operational meaning of parallel independence, but we are faced with
another problem. We are dealing with possibly infinite sets of
matchings of rules in a graph, and we cannot form a notion of infinite
sequences of rewrite steps (because each step may both remove and add
data). Yet we do not wish to restrict the notion to finite sets, not
simply for the sake of generality but also because it is closely
related to parallel independence, a notion that can naturally be
defined on infinite sets (see Section~\ref{sec-parindep}).

We may however use Definition \ref{def-parallel-rew} to handle
infinite sets of matchings, by using the graph $\Ipgr{G}{M}$ to stand
for the result of an (independent) sequence of transformations.  We
may thus express sequential independence as a generalized swapping
property, where the swap is performed between one transformation and
all the others (taken in parallel).  Yet this definition would not
imply that all subsets of a sequential independent set are sequential
independent, hence it needs to be stated in a more general way, by
swapping any transformation with any others (and not only with all the
others).

\begin{definition}\label{def-seqindep}[sequential independence]
  For any graph $G$ and set $M \subseteq\Matches{\R}{G}$, we say that
  $M$ is \emph{sequential independent} if for all $N\subseteq M$ and
  all $\mm\in M\setminus N$, 
  \begin{itemize}
  \item $\mm(\Lg{\mm})\subalg \Ipgr{G}{N}$, hence there is a canonical
    injection $j$ from $\mm(\Lg{\mm})$ to $\Ipgr{G}{N}$,
  \item there exists an isomorphism $\am$ such that
    $\am(\Ipgr{G}{N\cup\set{\mm}})=
    \Ipgr{\big(\Ipgr{G}{N}\big)}{j\circ\mm}$
    and $\am$ is the identity on $G$.
  \end{itemize}
\end{definition}

The isomorphism $\am$ in~Definition \ref{def-seqindep} is necessary to
account for the difference between the isomorphic graphs
$\liftm{\mm}(\Rg{\mm})$ and $\liftm{(j\circ \mm)}(\Rg{\mm})$, i.e., to
transform vertices or arrows of the form $\tuple{x,\mm}$ into
$\tuple{x,j\circ\mm}$ (but there is no need to be that specific in the
definition).

It is then easy to see (by induction on the cardinality of $M$) that
\begin{prop}\label{prop-seqindep}
  For any graph $G$ and finite set $M \subseteq\Matches{\R}{G}$, if $M$ is
  sequential independent then \[G \seqrew{\R} \Ipgr{G}{M}.\]
\end{prop}

Of course there is usually more than one sequence of rewriting steps
from $G$ to $\Ipgr{G}{M}$, since under the hypothesis they can be
swapped; but without it there is generally none (as illustrated in
Example~\ref{ex-Ipgr}). And the fact that there is one such sequence
does not imply sequential independence, i.e., the converse of
Proposition \ref{prop-seqindep} is obviously not true.

\section{Parallel Independence}\label{sec-parindep}

In the Double-Pushout approach, parallel independence is a property of
two direct transformations of the same object $G$, formulated as the
existence of two commuting morphisms $j_1$ and $j_2$ as shown below.

\begin{center}
  \begin{tikzpicture}[xscale=1.65, yscale=1.8]
    \node (L) at (0.5,1) {$L_2$}; \node (K) at (1.5,1) {$K_2$};
    \node (R) at (2.5,1) {$R_2$}; \node (G) at (0,0) {$G$};
    \node (D) at (1.5,0) {$D_2$}; \node (H) at (2.5,0) {$H_2$};
    \path[->] (K) edge (L);
    \path[->] (K) edge (R);
    \path[->] (L) edge node[fill=white, font=\footnotesize, near end] {$\nm$} (G);
    \path[->] (K) edge (D);
    \path[->] (D) edge (G);
    \path[->] (D) edge (H);
    \path[->] (R) edge (H);
    \node (L1) at (-0.5,1) {$L_1$}; \node (K1) at (-1.5,1) {$K_1$};
    \node (R1) at (-2.5,1) {$R_1$};
    \node (D1) at (-1.5,0) {$D_1$}; \node (H1) at (-2.5,0) {$H_1$};
    \path[->] (K1) edge (L1);
    \path[->] (K1) edge (R1);
    \path[->] (L1) edge node[fill=white, font=\footnotesize, near end] {$\mm$} (G);
    \path[->] (K1) edge (D1);
    \path[->] (D1) edge (G);
    \path[->] (D1) edge (H1);
    \path[->] (R1) edge (H1);
  \path[-] (L1) edge[draw=white, line width=3pt]  (D);
  \path[-] (L) edge[draw=white, line width=3pt]  (D1);
    \path[->,dashed] (L1) edge node[fill=white, font=\footnotesize] {$j_1$} (D);
    \path[->,dashed] (L) edge node[fill=white, font=\footnotesize] {$j_2$} (D1);
  \end{tikzpicture}    
\end{center}

The Local Church-Rosser Theorem mentioned above actually shows that
$\mm$ and $\nm$ are parallel independent iff they correspond to a
sequential independent pair $\tuple{\mm,\nm'}$ (where
$\nm':L_2\rightarrow H_1$ is related to $\nm$). It is the symmetry
between $\mm$ and $\nm$ that entails the swapping property. This is
remarkable since parallel independence does not refer to the
\emph{results} of the transformations involved, while the result of
the sequences of transformations is central in the swapping property
(as in Definition~\ref{def-seqindep}).

This definition of parallel independence can easily be lifted to sets
$M$ of matchings (or direct transformations) by considering all
possible pairs $\mm,\nm\in M$, with a slight caveat. In this
definition the two direct transformations may be identical, thus
stating a property of a single transformation that is not shared by
all direct transformations. But Definition~\ref{def-parallel-rew} does
not allow to apply any member $\mm$ of $M$ more than once (because
applying $\mm$ any number of times in parallel would jeopardize
determinism of $\fullPRr{\R}$, see Definition~\ref{def-EDP}
below). For this reason we will only consider pairs of distinct
matchings (so that singletons $M$ shall be considered as parallel
independent, see below).

Our goal is therefore to formulate parallel independence in the
present framework, in order to obtain an equivalence similar to the
Local Church-Rosser Theorem. Considering that the pushout complement
$D_1$ is replaced by the graph
$\delgraph{G}{\Vd{\mm}}{\Ad{\mm}}{\ld{\mm}}$, the commuting property
of $j_2$ amounts to
$\nm(L_2)\subalg \delgraph{G}{\Vd{\mm}}{\Ad{\mm}}{\ld{\mm}}$, that can
be more elegantly expressed as
$\nm(L_2) \gcap \mm(L_1) \subalg \mm(K_1)$, or
$\nm(\Lg{\nm}) \gcap \mm(\Lg{\mm}) \subalg \mm(\Kg{\mm})$ using our
notations.  This simply means that any graph item that is matched by
two concurrent rules cannot be removed. The commuting property of
$j_1$ is obtained by swapping $\mm$ and $\nm$.

However, our treatment of attributes makes it possible to recover in
the right-hand side an attribute that has been deleted in the
left-hand side (this is of course not possible for vertices or
arrows). This possibility should therefore be accounted for in the
notion of parallel independence, i.e., an attribute that is matched
twice may be deleted provided it is recovered. This can be expressed
as
\[\nm(\Lg{\nm}) \gcap \mm(\Lg{\mm}) \subalg \mm(\Kg{\mm}) \gcup
  \liftm{\mm}(\Rg{\mm}) \text{ for all } \mm,\nm\in M \text{ such that
  } \mm\neq \nm.\] However, this is not a sufficient condition for
sequential independence.

\begin{example}\label{ex-notseqindep}
  We consider the following graph and rules:
  \begin{align*}
    G &= \raisebox{-1.3ex}{\begin{tikzpicture}
    \node[draw,rounded corners] at (0,0)
    {$x\sepattr 0$};
  \end{tikzpicture}} \text{ where } \carrier{\algebrf{G}}=\set{0}\\
    r_1 &= (\,\raisebox{-1.3ex}{\begin{tikzpicture}
    \node[draw,rounded corners] at (0,0)
    {$x_1\sepattr 0$};
  \end{tikzpicture}},\, 
  \raisebox{-1ex}{\begin{tikzpicture}
    \node[draw,rounded corners] at (0,0)
    {$x_1$};
  \end{tikzpicture}},\,\raisebox{-1ex}{\begin{tikzpicture}
    \node[draw,rounded corners] at (0,0)
    {$x_1$};
  \end{tikzpicture}}\,)\\
    r_2 &=(\,\raisebox{-1ex}{\begin{tikzpicture}
    \node[draw,rounded corners] at (0,0)
    {$x_2$};
  \end{tikzpicture}},\, \raisebox{-1ex}{\begin{tikzpicture}
    \node[draw,rounded corners] at (0,0)
    {$x_2$};
  \end{tikzpicture}},\, 
  \raisebox{-1.3ex}{\begin{tikzpicture}
    \node[draw,rounded corners] at (0,0)
    {$x_2\sepattr 0$};
  \end{tikzpicture}}\,)
  \end{align*}
  There is a unique matching $\mm$ of $r_1$ (resp. $\nm$ of $r_2$) in
  $G$, given by $\mm(x_1)=\nm(x_2)=x$ and
  $\attrf{\mm}(0)=\attrf{\nm}(0)=0$. We see that $M=\set{\mm,\nm}$ is
  not sequential independent. Indeed, let $N=\set{\nm}$, then
  $G=\mm(\Lg{\mm})\subalg \Ipgr{G}{N}=G$ (hence $j$ is the identity
  morphism of $G$), but
  $\Ipgr{\big(\Ipgr{G}{N}\big)}{\mm} = \Ipgr{G}{\mm} =
  \raisebox{-0.7ex}{\begin{tikzpicture} \node[draw,rounded corners] at
      (0,0) {$x$};
    \end{tikzpicture}}$ is not isomorphic to $\Ipgr{G}{M}=G$.

  Yet we see that
  \begin{align*}
    \nm(\Lg{\nm}) \gcap \mm(\Lg{\mm}) = \raisebox{-0.7ex}{\begin{tikzpicture} \node[draw,rounded corners] at
      (0,0) {$x$};
    \end{tikzpicture}} &\subalg \raisebox{-0.7ex}{\begin{tikzpicture} \node[draw,rounded corners] at
      (0,0) {$x$};
    \end{tikzpicture}} =  \mm(\Kg{\mm}) \gcup
                         \liftm{\mm}(\Rg{\mm}) \\
    \mm(\Lg{\mm}) \gcap \nm(\Lg{\nm}) = \raisebox{-0.7ex}{\begin{tikzpicture} \node[draw,rounded corners] at
      (0,0) {$x$};
    \end{tikzpicture}} &\subalg \raisebox{-1.3ex}{\begin{tikzpicture}
    \node[draw,rounded corners] at (0,0)
    {$x\sepattr 0$};
  \end{tikzpicture}} = \nm(\Kg{\nm}) \gcup \liftm{\nm}(\Rg{\nm}),
  \end{align*}
  which proves that this condition is true
  for all pairs of distinct elements of $M$; hence it is not
  sufficient to ensure sequential independence.
\end{example}
 
The problem in Example~\ref{ex-notseqindep} is that the attribute 0 of
$x$ is considered as being matched only once (by $\Lg{\mm}$), while it
is actually also matched by $\Rg{\nm}$.  This leads to the following
definition.

\begin{definition}[parallel independence]\label{def-parindep}
  For any graph $G$ and set $M \subseteq\Matches{\R}{G}$, we say that
  $M$ is \emph{parallel independent} if
  \[(\nm(\Lg{\nm})\gcup \liftm{\nm}(\Rg{\nm}))\gcap \mm(\Lg{\mm})\
    \subalg\ \mm(\Kg{\mm}) \gcup \liftm{\mm}(\Rg{\mm}) \text{\ \ for all } \mm,\nm\in M
    \text{ such that } \mm\neq \nm.\]
\end{definition}

This definition may seem strange, but it is easy to see that on
unlabeled graphs it amounts to
$\nm(\Lg{\nm})\gcap \mm(\Lg{\mm})\subalg \mm(\Kg{\mm})$ for all
$\mm\neq\nm$, i.e., to the standard algebraic notion of parallel
independence (translated to the present framework). 

It turns out that Definition~\ref{def-parindep} provides the expected
characterization of sequential independence.

\begin{theorem}\label{thm-para}
  For any graph $G$ and set $M\subseteq\Matches{\R}{G}$, $M$ is parallel
  independent iff $M$ is sequential independent.
\end{theorem}

The (rather long) proof of Theorem~\ref{thm-para} can be found in
\cite{BdlTE20c}.

We therefore see that Definition~\ref{def-parindep} arises as a
characterization of sequential independence that does not refer to the
results of the transformations, and indeed that does not rely on the
definition of $\Ipgr{G}{M}$ (Definition~\ref{def-parallel-rew}),
though of course it does rely on the definitions of unions of graphs,
of rules and of the matchings $\liftm{\mm}$
(Definitions~\ref{def-joinableg}, \ref{def-rules} and \ref{def-RImg}).
Note also that Definition~\ref{def-parindep} depends explicitly on the
right-hand sides of rules, in contrast with the general algebraic
definition of parallel independence given above, or with the Essential
Condition of parallel independence in \cite{CorradiniDLRMCA18}.

\section{Parallel Rewriting}\label{sec-EDP}

We have not yet defined a relation of parallel rewriting as we did for
sequential rewriting (Definition~\ref{def-seqrew}). The reason is that
two matchings may conflict as one retains (in $R\gcap K$) what another
removes. 

\begin{example}\label{ex-conflict}
  We consider the following unlabeled rule $r$ and graph $G$.
    \[r\ =\ (\,\raisebox{-3.9ex}{\begin{tikzpicture}
    \node[draw,rounded corners] at (0,0)
    {\begin{tikzpicture}[scale=2.2]
        \node (x) at (0,0) {$x$};
        \node (y) at (1,0) {$x'$};
        \path[->,bend left] (x) edge
        node[fill=white,font=\footnotesize] {$f$} (y);
        \path[->,bend left] (y) edge node[fill=white,font=\footnotesize] {$f'$} (x);
      \end{tikzpicture}};
  \end{tikzpicture}},\, 
  \raisebox{-0.7ex}{\begin{tikzpicture}
    \node[draw,rounded corners] at (0,0)
    {$x$};
  \end{tikzpicture}},\,\raisebox{-0.7ex}{\begin{tikzpicture}
    \node[draw,rounded corners] at (0,0)
    {$x$};
  \end{tikzpicture}}\,)\hspace{2em}
    G\ =\ \raisebox{-3.9ex}{\begin{tikzpicture}
    \node[draw,rounded corners] at (0,0)
    {\begin{tikzpicture}[scale=2.2]
        \node (x) at (0,0) {$y$};
        \node (y) at (1,0) {$z$};
        \path[->,bend left] (x) edge
        node[fill=white,font=\footnotesize] {$g$} (y);
        \path[->,bend left] (y) edge node[fill=white,font=\footnotesize] {$h$} (x);
      \end{tikzpicture}};
  \end{tikzpicture}} \]

  There are two matchings $\mm_1,\mm_2$ of $r$ in
  $G$, given by 
  \[\begin{array}{c|cccc}
      &x&x'&f&f'\\ \hline
      \mm_1&y&z&g&h
    \end{array}\hspace{2em} 
    \begin{array}{c|cccccc}
      &x&x'&f&f'\\ \hline
      \mm_2&z&y&h&g
    \end{array}\]

    According to rule $r$ with matching $\mm_1$, the node
    $\mm_1(x')=z$ and the arrows $\mm_1(f)=g$ and $\mm_1(f')=h$ have
    to be removed, and the node $\mm_1(x)=y$ should occur in the
    result of the transformation. But with matching $\mm_2$, the node $\mm_2(x')=y$ should
    be removed and the node $\mm_2(x)=z$ should be preserved. There is
    a conflict between $\mm_1$ and $\mm_2$ on the nodes of $G$ (but
    not on its arrows).

    Let $M=\set{\mm_1,\mm_2}$, then $\Vd{M}= \mm_1(\set{x'})\cup\mm_2(\set{x'}) =
    \set{y,z}=\nodes{G}$ hence $\delgraph{G}{\Vd{M}}{\Ad{M}}{\ld{M}}$
    is empty and
    \[\Ipgr{G}{M} = \mm_1(\raisebox{-0.7ex}{\begin{tikzpicture}
        \node[draw,rounded corners] at (0,0) {$x$};
  \end{tikzpicture}}) \gcup \mm_2(\raisebox{-0.7ex}{\begin{tikzpicture}
    \node[draw,rounded corners] at (0,0)
    {$x$};
  \end{tikzpicture}}) = \raisebox{-0.9ex}{\begin{tikzpicture}
    \node[draw,rounded corners] at (0,0)
    {$y$};
  \end{tikzpicture}} \gcup \raisebox{-0.7ex}{\begin{tikzpicture}
    \node[draw,rounded corners] at (0,0)
    {$z$};
  \end{tikzpicture}} = \raisebox{-1.8ex}{\begin{tikzpicture}
    \node[draw,rounded corners] at (0,0)
    {\begin{tikzpicture}[scale=1]
        \node (x) at (0,0) {$y$};
        \node (y) at (1,0) {$z$};
              \end{tikzpicture}};
  \end{tikzpicture}}\]

\end{example}

The transformation offered by Definition~\ref{def-parallel-rew}
performs deletions before unions, which means that these conflicts are
resolved by giving priority to retainers over removers. But if the
deletion actions of a rule are not executed in a parallel
transformation, how can we claim that this rule has been executed (or
applied) in parallel with others? Thus, in order to define parallel
rewriting with a clear semantics we need to rule out such conflicts.

A natural restriction is therefore to make sure that the items that
should be removed, i.e., those contained in $\Vd{M}$, $\Ad{M}$ or
$\ld{M}$, have indeed been removed from the result.

\begin{definition}[regularity]\label{def-regular}
  For any graph $G$ and set $M \subseteq\Matches{\R}{G}$, we say that
  $M$ is \emph{regular} if $\Ipgr{G}{M}$ is disjoint from
  $\Vd{M}, \Ad{M}, \ld{M}$.
\end{definition}

As for sequential independence, this property of $M$ can be
characterized as a property of pairs of elements of $M$.

\begin{lemma}\label{lm-regular}
  For any graph $G$ and set $M \subseteq\Matches{\R}{G}$, \[M \text{ is
  regular\ \  iff\ \ }
  \liftm{\nm}(\Rg{\nm}) \gcap \mm(\Lg{\mm})\subalg \mm(\Kg{\mm})\text{ for
  all }\mm,\nm\in M.\]
\end{lemma}
\begin{proof}
  Let $H = \Gcup_{\nm\in M}\RImg{G}{\nm}$, then $\Ipgr{G}{M} =
  \delgraph{G}{\Vd{M}}{\Ad{M}}{\ld{M}} \gcup H$ is disjoint from
  $\Vd{M}, \Ad{M}, \ld{M}$ iff $H$ is. We have
\[\nodes{H}\cap \Vd{M} = \Big(\bigcup_{\nm\in M}\RImg{\nodes{G}}{\nm}\Big) \cap
  \Big(\bigcup_{\mm\in M}\mm(\nodesLg{\mm}\setminus
  \nodesKg{\mm})\Big) = \bigcup_{\mm,\nm\in M}\liftm{\nm}(\nodesRg{\nm})
  \cap \mm(\nodesLg{\mm})\setminus \mm(\nodesKg{\mm})\] hence
$\nodes{H}\cap \Vd{M} = \ensvide$ iff
$\liftm{\nm}(\nodesRg{\nm}) \cap \mm(\nodesLg{\mm})\setminus \mm(\nodesKg{\mm}) =
\ensvide$ for all $\mm,\nm\in M$, but this is equivalent to
$\liftm{\nm}(\nodesRg{\nm}) \cap \mm(\nodesLg{\mm})\subseteq
\mm(\nodesKg{\mm})$. Similarly we see that
$\arrows{H}\cap \Ad{M} = \ensvide$ iff
$\liftm{\nm}(\arrowsRg{\nm}) \cap \mm(\arrowsLg{\mm})\subseteq
\mm(\arrowsKg{\mm})$ for all $\mm,\nm\in M$. 

For every vertex or arrow $x$ of $\Ipgr{G}{M}$ we have
\begin{align*}
  \labelf{H}(x)\cap \ld{M}(x)
  &= \bigcup_{\nm\in M} \labelf{\nm}\circ \labelfRg{\nm} \circ
    \invf{\nm}(x) \cap \ld{M}(x)\\
  &= \bigcup_{\mm,\nm\in M} \labelf{\nm}\circ \labelfRg{\nm} \circ
    \invf{\nm}(x) \cap \labelf{\mm}\circ (\labelfLg{\mm}\setminus \labelfKg{\mm})\circ \invf{\mm}(x)\\
  &= \bigcup_{\mm\nm\in M} \labelf{\nm}\circ \labelfRg{\nm} \circ
    \invf{\nm}(x) \cap \labelf{\mm}\circ \labelfLg{\mm}\circ
    \invf{\mm}(x) \setminus \labelf{\mm}\circ \labelfKg{\mm}\circ \invf{\mm}(x)
\end{align*}
by using the fact that $\mm$ is consistent. We therefore see that
$\labelf{H}(x)\cap \ld{M}(x)=\ensvide$ holds iff
$\labelf{\nm}\circ \labelfRg{\nm} \circ \invf{\nm}(x)
\cap \labelf{\mm}\circ \labelfLg{\mm}\circ \invf{\mm}(x)
\subseteq \labelf{\mm}\circ \labelfKg{\mm}\circ \invf{\mm}(x)$ holds
for all $\mm,\nm\in M$. By definition $M$ is regular iff
$\nodes{H}\cap \Vd{M} =\arrows{H}\cap \Ad{M} = \ensvide$ and
$\labelf{H}\cap \ld{M}$ is empty everywhere, hence $M$ is regular iff
$\liftm{\nm}(\Rg{\nm}) \gcap \mm(\Lg{\mm})\subalg \mm(\Kg{\mm})$ for
all $\mm,\nm\in M$.
\end{proof}
\begin{corollary}\label{cr-regular}
  $M$ is regular iff all its subsets are regular.
\end{corollary}

These nice properties, and the fact that regularity ensures the
absence of conflicts, are however not sufficient in the light of
parallel independence. Indeed, we now show that a parallel independent
set may not be regular.

\begin{example}\label{ex-EDP}
  Let us consider rules $r_1=\tuple{L_1,K_1,R_1}$ and
  $r_2=\tuple{L_2,K_2,R_2}$ where the graphs $L_1$, $K_1$ and $R_1$
  have only one vertex $x_1$, the graphs $L_2$, $K_2$ and $R_2$ have
  only one vertex $x_2$, and the attributes are as pictured below ($u,v$
  are variables and $f$ is a unary function symbol). Let $\algebrf{G}$ be the
  algebra with carrier set $\set{0}$ where $f$ is interpreted as
  the constant function $0$, and let $G$ be the graph that has
  a unique vertex $x$ with attribute $0$.
  
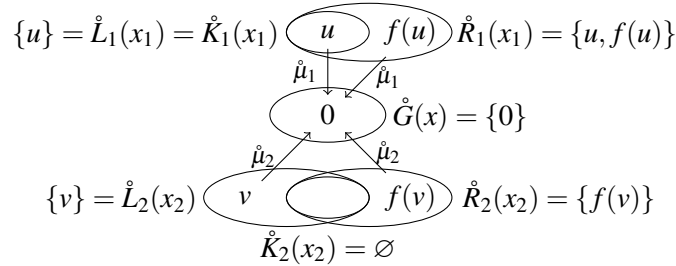
\begin{center}
    \begin{tikzpicture}[scale=1.1]
      \draw (0,1) circle (.5cm and .25cm);
      \draw (0.5,1) circle (1cm and .35cm);
      \draw (0.5,-1) circle (1cm and .35cm);
      \draw (-0.5,-1) circle (1cm and .35cm);
      \draw (0,-1) circle (.5cm and .25cm);
      \draw (0,0) circle (0.7cm and .33cm);
      \node (G) at (0,0) {$0$}; \node (U) at (0,1) {$u$};
      \node (FU) at (1,1) {$f(u)$}; \node (V) at (-1,-1) {$v$};
      \node (FV) at (1,-1) {$f(v)$};
      \node at (-2.2,1) {$\set{u}=\labelf{L}_1(x_1)=\labelf{K}_1(x_1)$};
      \node at (2.9,1) {$\labelf{R}_1(x_1) = \set{u,f(u)}$};
      \node at (1.6,0) {$\labelf{G}(x)=\set{0}$};
      \node at (-2.5,-1) {$\set{v}=\labelf{L}_2(x_2)$};
      \node at (0,-1.6) {$\labelf{K}_2(x_2)=\ensvide$};
      \node at (2.8,-1) {$\labelf{R}_2(x_2)=\set{f(v)}$};
      \path[->] (U) edge node[left, font=\footnotesize] {$\labelf{\mm}_1$} (G);
      \path[->] (FU) edge node[right, font=\footnotesize] {$\labelf{\mm}_1$} (G);
      \path[->] (V) edge node[left, font=\footnotesize] {$\labelf{\mm}_2$} (G);
      \path[->] (FV) edge node[right, font=\footnotesize] {$\labelf{\mm}_2$} (G);
    \end{tikzpicture}
  \end{center}

  There are exactly two matchings of $\set{r_1,r_2}$ in $G$: $\mm_1$
  and $\mm_2$ defined by $\mm_1(x_1)=\mm_2(x_2)=x$ and
  $\labelf{\mm}_1(u)=\labelf{\mm}_2(v)=0$. Let $M=\set{\mm_1,\mm_2}$,
  we see by Lemma~\ref{lm-regular} that $M$ is not regular since
  \[\liftm{\mm_1}(R_1)\gcap\mm_2(L_2) = \mm_1(\,\raisebox{-1.3ex}{\begin{tikzpicture}
    \node[draw,rounded corners] at (0,0)
    {$x_1\sepattr u,f(u)$};
  \end{tikzpicture}}\,) \gcup \mm_2(\,\raisebox{-1.3ex}{\begin{tikzpicture}
    \node[draw,rounded corners] at (0,0)
    {$x_2\sepattr v$};
  \end{tikzpicture}}\,) = \raisebox{-1.2ex}{\begin{tikzpicture}
    \node[draw,rounded corners] at (0,0)
    {$x\sepattr 0$};
  \end{tikzpicture}} = G\] is not a subgraph of
  $\mm_2(K_2) = \mm_2(\raisebox{-0.9ex}{\begin{tikzpicture}
    \node[draw,rounded corners] at (0,0)
    {$x_2$};
  \end{tikzpicture}}) = \raisebox{-0.7ex}{\begin{tikzpicture}
    \node[draw,rounded corners] at (0,0)
    {$x$};
  \end{tikzpicture}}$ (or equivalently because $\Ipgr{G}{M}=G$ is
not disjoint from $\ld{M}$). 

  However, we see that $M$ is sequential independent since the
  matchings can be applied sequentially in any order, yielding in both
  cases the graph $G$. Equivalently, $M$ is parallel independent since
  \begin{align*}
    (\mm_2(L_2)\gcup \liftm{\mm_2}(R_2))\gcap \mm_1(L_1) = 
    (G\gcup G)\gcap G = G &\subalg G\gcup G = \mm_1(K_1)\gcup \liftm{\mm_1}(R_1)\\
    (\mm_1(L_1)\gcup \liftm{\mm_1}(R_1))\gcap \mm_2(L_2) = 
    (G\gcup G)\gcap G = G &\subalg \raisebox{-0.7ex}{\begin{tikzpicture}
    \node[draw,rounded corners] at (0,0)
    {$x$};
  \end{tikzpicture}} \gcup G = \mm_2(K_2)\gcup \liftm{\mm_2}(R_2).
  \end{align*}
\end{example}

Note that conversely a set may be regular and not parallel
independent, as is the case of the set $M$ in Example~\ref{ex-Ipgr}.

We obviously need a more comprehensive notion of parallel rewriting,
one that applies \emph{at least} on all parallel independent sets of
matchings. We see in Example~\ref{ex-EDP} that the two rules do clash
on the attribute 0 of $x$, but the clash is settled by their
right-hand sides. This leads to the following definition from
\cite{BdlTE20c}.

\begin{definition}[effective deletion property, parallel rewriting]\label{def-EDP}
  For any graph $G$, a set $M \subseteq\Matches{\R}{G}$ is said to
  satisfy the \emph{effective deletion property} if $\Ipgr{G}{M}$
  is disjoint from $\Vd{M}, \Ad{M}, \ld{M}\setminus \llift{M}$, where
  \[\llift{M}\defeq \bigcup_{\mm\in
    M}\attrf{\mm}\circ(\attrfRg{\mm}\setminus
  \attrfKg{\mm})\circ\invf{\mm}.\]

For any finite set of rules $\R$, we define the relation $\IPTrew{\R}$
of \emph{parallel rewriting} by stating that, for all graphs $G$ and
$H$, $G\IPTrew{\R} H$ iff there exists a set
$M\subseteq\Matches{\R}{G}$ that has the effective deletion property
and such that $H\isog \Ipgr{G}{M}$. We write $G\fullPRr{\R} H$ if
$M=\Matches{\R}{G}$.
\end{definition}

The effective deletion property is obviously more general than
regularity. The example below shows that it is strictly more general
than regularity.

\begin{example}\label{ex-reg-parindep-to-edp}
  We consider again Example~\ref{ex-notseqindep} where
  $M=\set{\mm,\nm}$ is not sequential independent, hence by
  Theorem~\ref{thm-para} $M$ is not parallel independent. We have
  $\Vd{M}=\Ad{M}=\emptyset$ and $\ld{M}(x)=\set{0}$. Since
  $\Ipgr{G}{M}=\raisebox{-1.3ex}{\begin{tikzpicture}
      \node[draw,rounded corners] at (0,0)
      {$x\sepattr 0$};
    \end{tikzpicture}}$ is not disjoint from $\Vd{M},\Ad{M},\ld{M}$ then
  $M$ is not regular. But
  \[\llift{M}(x) = \attrf{\mm}\circ(\attrf{R}_1\setminus
    \attrf{K}_1)\circ\invf{\mm}(x)\ \cup\ 
    \attrf{\nm}\circ(\attrf{R_2}\setminus
    \attrf{K_2})\circ\invf{\nm}(x) = \set{0},\] hence
  $\ld{M}(x)\setminus \llift{M}(x)=\ensvide$ and therefore $M$ has the
  effective deletion property.
\end{example}

It has been shown in \cite{BdlTE20c} that $\fullPRr{\R}$ is
deterministic up to isomorphism, that is, if $G\fullPRr{\R} H$,
$G'\fullPRr{\R} H'$ and $G\isog G'$ then $H\isog H'$. In particular,
it is possible to represent any cellular automaton by a suitable rule
$r$ and a class of graphs that correspond to configurations of the
automaton (every vertex corresponds to a cell), such that
$\fullPRr{r}$ (restricted to such graphs) is the transition function
of the automaton. Furthermore, it is proved in \cite{BdlTE20c} (as a
lemma to Theorem~\ref{thm-para}) that

\begin{theorem}\label{thm-indep2edp}
  For any graph $G$ and set $M\subseteq\Matches{\R}{G}$ if $M$ is
  parallel independent then $M$ has the effective
  deletion property.
\end{theorem}

Hence effective deletion supports a definition of parallel rewriting
that is general enough to handle parallel independence. Besides,
Example~\ref{ex-reg-parindep-to-edp} also shows that the effective
deletion property is strictly more general than parallel independence.

We also see that
\begin{corollary}
  If $M\subseteq \Matches{\R}{G}$ is finite and parallel independent then $G
  \seqrew{\R} \Ipgr{G}{M}$ and $G\IPTrew{\R} \Ipgr{G}{M}$.
\end{corollary}
\begin{proof}
  By Theorem~\ref{thm-indep2edp} we have $G\IPTrew{\R}
  \Ipgr{G}{M}$. By Theorem~\ref{thm-para} $M$ is sequential
  independent, hence by Proposition~\ref{prop-seqindep} we have $G
  \seqrew{\R} \Ipgr{G}{M}$. 
\end{proof}
Hence in this case parallel and sequential rewriting meet, and
parallel rewriting can be said to yield a \emph{correct} result
  w.r.t. sequential rewriting.

\section{Parallel Coherence}\label{sec-parcoh}

One drawback of the effective deletion property is that it cannot be
characterized as a property of pairs of elements of $M$, as the
following example shows.

\begin{example}
  We consider the following graph and rule
  \begin{align*}
    G &= \raisebox{-1.3ex}{\begin{tikzpicture}
    \node[draw,rounded corners] at (0,0)
    {$x\sepattr 0$};
  \end{tikzpicture}} \text{ where } \carrier{\algebrf{G}}=\set{0,1}\\
    r_3 &= (\,\raisebox{-1.3ex}{\begin{tikzpicture}
        \node[draw,rounded corners] at (0,0)
        {$x_3\sepattr 0$};
      \end{tikzpicture}},\, \raisebox{-1.3ex}{\begin{tikzpicture}
        \node[draw,rounded corners] at (0,0)
        {$x_3\sepattr 0$};
      \end{tikzpicture}},\, 
    \raisebox{-1.3ex}{\begin{tikzpicture}
        \node[draw,rounded corners] at (0,0)
        {$x_3\sepattr 0,1$};
      \end{tikzpicture}}\,)
  \end{align*}
  and also the rules $r_1$, $r_2$ of Example~\ref{ex-notseqindep}. For
  $i=1,2,3$ let $\mm_i$ be the unique matching of $r_i$ in $G$ such
  that $\attrf{\mm}_i$ is the identity function of $\set{0,1}$. Let
  $M=\set{\mm_1,\mm_2,\mm_3}$ and $N=\set{\mm_1,\mm_3}$. We obviously
  have $\Vd{M}=\Vd{N}=\Ad{M}=\Ad{N}=\ensvide$. We see that
  $\ld{\mm_1}(x) = \set{0}$ and $\ld{\mm_2}(x) = \ld{\mm_3}(x) =
  \ensvide$, so that $\ld{N}(x)=\ld{M}(x)=\set{0}$,
  \[\Ipgr{G}{N} = \raisebox{-0.7ex}{\begin{tikzpicture}
        \node[draw,rounded corners] at (0,0) {$x$};
      \end{tikzpicture}} \gcup \raisebox{-0.7ex}{\begin{tikzpicture}
        \node[draw,rounded corners] at (0,0) {$x$};
      \end{tikzpicture}} \gcup \raisebox{-1.3ex}{\begin{tikzpicture}
        \node[draw,rounded corners] at (0,0) {$x\sepattr 0,1$};
      \end{tikzpicture}} = \raisebox{-1.3ex}{\begin{tikzpicture}
        \node[draw,rounded corners] at (0,0) {$x\sepattr 0,1$};
      \end{tikzpicture}}\ \text{ and }\  \Ipgr{G}{M} = \raisebox{-0.7ex}{\begin{tikzpicture}
        \node[draw,rounded corners] at (0,0) {$x$};
      \end{tikzpicture}} \gcup \raisebox{-0.7ex}{\begin{tikzpicture}
        \node[draw,rounded corners] at (0,0) {$x$};
      \end{tikzpicture}} \gcup \raisebox{-1.3ex}{\begin{tikzpicture}
        \node[draw,rounded corners] at (0,0) {$x\sepattr 0$};
      \end{tikzpicture}} \gcup \raisebox{-1.3ex}{\begin{tikzpicture}
        \node[draw,rounded corners] at (0,0) {$x\sepattr 0,1$};
      \end{tikzpicture}} = \raisebox{-1.3ex}{\begin{tikzpicture}
        \node[draw,rounded corners] at (0,0) {$x\sepattr 0,1$};
      \end{tikzpicture}}\] We then see that
  $\llift{\mm_1}(x)=\ensvide$, $\llift{\mm_2}(x)=\set{0}$ and
  $\llift{\mm_3}(x)=\set{1}$, so that $\llift{N}(x)=\set{1}$ and
  $\llift{M}(x)=\set{0,1}$. Hence
  $\ld{M}(x)\setminus\llift{M}(x)=\ensvide$ and
  $\ld{N}(x)\setminus\llift{N}(x)=\set{0}$, and $M$ but not $N$ has
  the effective deletion property.
\end{example}

The reader may find strange that the conflict between $r_1$ and $r_3$
could be settled by some other rule, here $r_2$. This means that we
need the whole of $M$ to decide wether all conflicts are settled. For
this reason the effective deletion property may appear as too general.

Another possibility for defining parallel rewriting is to translate to
the present framework the notion of parallel coherence that has
been devised in order to define algebraic parallel graph
transformation (see \cite{BdlTE20b}). In that paper we used production
rules of the form $L\leftarrow K\leftarrow I \rightarrow R$ that do
not require a morphism from $K$ to $R$. Direct derivations are commuting
diagrams
\begin{center}
  \begin{tikzpicture}[xscale=1.8, yscale=1.4]
    \node (L) at (0,1) {$L$}; \node (K) at (1,1) {$K$};  \node (I) at
    (2,1) {$I$};  \node (R) at (3,1) {$R$}; \node (G) at (0.5,0) {$G$};
    \node (D) at (1.5,0) {$D$}; \node (H) at (2.5,0) {$H$};
    \path[->] (K) edge (L);
    \path[->] (I) edge (K);
    \path[->] (I) edge (R);
    \path[->] (L) edge (G);
    \path[->] (K) edge (D);
    \path[->] (D) edge (G);
    \path[->] (I) edge (D);
    \path[->] (D) edge (H);
    \path[->] (R) edge (H);
  \end{tikzpicture}
\end{center}
where the squares are pushouts. Note that a standard Double-Pushout
can be obtained with $K=I$.  \emph{Parallel coherence}, as a property
of two direct transformations of the same object $G$, is defined as
the existence of two commuting morphisms $j_1$ and $j_2$ as shown
below.

\begin{center}
  \begin{tikzpicture}[xscale=1.65, yscale=2]
    \node (L) at (0.5,1) {$L_2$}; \node (K) at (1.5,1) {$K_2$};  \node (I) at
    (2.5,1) {$I_2$};  \node (R) at (3.5,1) {$R_2$}; \node (G) at (0,0) {$G$};
    \node (D) at (2,0) {$D_2$}; \node (H) at (3,0) {$H_2$};
    \path[->] (K) edge (L);
    \path[->] (I) edge (K);
    \path[->] (I) edge (R);
    \path[->] (L) edge node[fill=white, font=\footnotesize, near start] {$\nm$} (G);
    \path[->] (K) edge (D);
    \path[->] (D) edge (G);
    \path[->] (I) edge (D);
    \path[->] (D) edge (H);
    \path[->] (R) edge (H);
    \node (L1) at (-0.5,1) {$L_1$}; \node (K1) at (-1.5,1) {$K_1$};  \node (I1) at
    (-2.5,1) {$I_1$};  \node (R1) at (-3.5,1) {$R_1$};
    \node (D1) at (-2,0) {$D_1$}; \node (H1) at (-3,0) {$H_1$};
    \path[->] (K1) edge (L1);
    \path[->] (I1) edge (K1);
    \path[->] (I1) edge (R1);
    \path[->] (L1) edge node[fill=white, font=\footnotesize, near start] {$\mm$} (G);
    \path[->] (K1) edge (D1);
    \path[->] (D1) edge (G);
    \path[->] (I1) edge (D1);
    \path[->] (D1) edge (H1);
    \path[->] (R1) edge (H1);
    \path[-] (I1) edge[draw=white, line width=3pt]  (D);
    \path[-] (I) edge[draw=white, line width=3pt]  (D1);
    \path[->,dashed] (I1) edge node[fill=white, font=\footnotesize,
    near start] {$j_1$} (D);
    \path[->,dashed] (I) edge node[fill=white, font=\footnotesize,
    near start] {$j_2$} (D1);
  \end{tikzpicture}    
\end{center}

This notion clearly generalizes algebraic parallel independence and is
therefore a good candidate. In the present framework the object $I_2$
is replaced by the graph $R_2\gcap K_2$, hence the commuting property
of $j_2$ amounts to
$\nm(R_2\gcap K_2)\subalg \delgraph{G}{\Vd{\mm}}{\Ad{\mm}}{\ld{\mm}}$,
that can be expressed as
$\mm(\Lg{\mm}) \gcap \nm(\Rg{\nm}\gcap \Kg{\nm}) \subalg
\mm(\Kg{\mm})$.  This simply means that any graph item that is matched
by some $R\gcap K$ cannot be removed by any rule.

\begin{definition}[parallel coherence]\label{def-parcoh}
  For any graph $G$ and set $M \subseteq\Matches{\R}{G}$, we say that
  $M$ is \emph{parallel coherent} if
  \[\nm(\Rg{\nm}\gcap \Kg{\nm})\gcap \mm(\Lg{\mm}) \subalg
  \mm(\Kg{\mm}) \text{ for all } \mm,\nm\in M.\]
\end{definition}

We easily show that this notion is more general than regularity.

\begin{lemma}\label{lm-reg2parcoh}
  For any graph $G$ and set $M \subseteq\Matches{\R}{G}$, if $M$ is
  regular then $M$ is parallel coherent.
\end{lemma}
\begin{proof}
  By Lemma~\ref{lm-regular} we have
  $\liftm{\nm}(\Rg{\nm}) \gcap \mm(\Lg{\mm})\subalg \mm(\Kg{\mm})$ for
  all $\mm,\nm\in M$. Since $\Rg{\nm}\gcap \Kg{\nm} \subalg \Rg{\nm}$
  then $\nm(\Rg{\nm}\gcap \Kg{\nm}) = \liftm{\nm}(\Rg{\nm}\gcap
  \Kg{\nm}) \subalg \liftm{\nm}(\Rg{\nm})$, hence $\nm(\Rg{\nm}\gcap
  \Kg{\nm})\gcap \mm(\Lg{\mm}) \subalg \liftm{\nm}(\Rg{\nm}) \gcap
  \mm(\Lg{\mm}) \subalg \mm(\Kg{\mm})$, hence $M$ is parallel coherent.
\end{proof}

It is easy to see that the converse does not hold (use for instance
Example~\ref{ex-notseqindep}). We now show that parallel coherence is
a restriction of the (possibly too general) effective deletion
property.

\begin{theorem}\label{thm-parcoh2edp}
  For any graph $G$ and set $M\subseteq\Matches{\R}{G}$, if $M$ is
  parallel coherent then $M$ has the effective
  deletion property.
\end{theorem}
\begin{proof}
  Let $H=\Ipgr{G}{M}$ then as in the proof of Lemma~\ref{lm-regular} we
  have
  \[\nodes{H}\cap \Vd{M} = \bigcup_{\mm,\nm\in M}\liftm{\nm}(\nodesRg{\nm})
    \cap \mm(\nodesLg{\mm})\setminus \mm(\nodesKg{\mm}).\] 
  But $\mm(\nodesLg{\mm})\subseteq \nodes{G}$ and by
  Definition~\ref{def-RImg} we have
  \begin{align*}
    \nodes{G}\cap \liftm{\nm}(\nodesRg{\nm}) 
    &= \nodes{G}\cap \RImg{\nodes{G}}{\nm} \\
    &= \nodes{G} \cap \big({\nm}(\nodesRg{\nm}\cap \nodesKg{\nm})\cup
      ((\nodesRg{\nm} \setminus \nodesKg{\nm}) \times \set{\nm})\\
    &= {\nm}(\nodesRg{\nm}\cap \nodesKg{\nm}),
  \end{align*}
  hence
    \[\nodes{H}\cap \Vd{M} 
    =  \bigcup_{\mm,\nm\in M}
          \nm(\nodesRg{\nm}\cap\nodesKg{\nm})\cap
          \mm(\nodesLg{\mm})\setminus \mm(\nodesKg{\mm})
     =  \ensvide\]
  since by parallel coherence
  $\nm(\nodesRg{\nm}\cap\nodesKg{\nm})\cap \mm(\nodesLg{\mm})\subseteq
  \mm(\nodesKg{\mm})$ for all $\mm,\nm\in M$. Similarly
  $\arrows{H}\cap \Ad{M} = \ensvide$.

  For all $x\in \nodes{H}\cup \arrows{H}$,
  if $x\not\in \nodes{G}\cup\arrows{G}$ then $\ld{M}(x)=\ensvide$ and
  obviously $\labelf{H}(x)\cap\ld{M}(x)\setminus
  \llift{M}(x)=\ensvide$. Otherwise $x\in \nodes{G}\cup\arrows{G}$ hence
  $\invf{\liftm{\nm}}(x)=\invf{\nm}(x)$ so that
  \[\labelf{H}(x) = \big(\labelf{G}(x)\setminus
    \ld{M}(x)\big)\cup\bigcup_{\nm\in
      M}\labelf{\nm}\circ \labelfRg{\nm}\circ \invf{\nm}(x).\] 
  Using the identity $A=(A\setminus B)\cup (A\cap B)$ for all sets $A$
  and $B$ we have
  \begin{align*}
    \labelf{\nm}\circ\labelfRg{\nm}\circ \invf{\nm}(x)
    &= \big(\labelf{\nm}\circ\labelfRg{\nm}\circ \invf{\nm}(x) \setminus \labelf{\nm}\circ
      (\labelfRg{\nm}\cap \labelfKg{\nm})\circ \invf{\nm}(x)\big)\\
    & \quad \cup
      \big(\labelf{\nm}\circ\labelfRg{\nm}\circ \invf{\nm}(x) \cap \labelf{\nm}\circ
      (\labelfRg{\nm}\cap \labelfKg{\nm})\circ \invf{\nm}(x)\big) 
  \end{align*}
  for all $\nm\in M$. By parallel coherence we have
  $\labelf{\nm}\circ (\labelfRg{\nm}\cap \labelfKg{\nm})\circ
  \invf{\nm}(x)\cap \labelf{\mm}\circ
  \labelfLg{\mm}\circ \invf{\mm}(x) \subseteq \labelf{\mm}\circ
  \labelfKg{\mm}\circ \invf{\mm}(x)$  for all $\mm,\nm\in M$, and
  since $\mm$ is consistent we get
  \begin{align*}
      \labelf{\nm}\circ (\labelfRg{\nm}\cap \labelfKg{\nm})\circ \invf{\nm}(x)\cap \ld{M}(x)
    & = \bigcup_{\mm\in M}\labelf{\nm}\circ (\labelfRg{\nm}\cap \labelfKg{\nm})\circ
    \invf{\nm}(x)\cap \labelf{\mm}\circ (\labelfLg{\mm}\setminus \labelfKg{\mm})\circ \invf{\mm}(x)\\
    & = \bigcup_{\mm\in M}\labelf{\nm}\circ (\labelfRg{\nm}\cap \labelfKg{\nm})\circ
    \invf{\nm}(x)\cap \labelf{\mm}\circ \labelfLg{\mm}\circ
      \invf{\mm}(x)  \setminus \labelf{\mm}\circ\labelfKg{\mm}\circ \invf{\mm}(x) \\ 
    & = \ensvide,
  \end{align*}
  hence
  \begin{align*}
    \labelf{H}(x)\cap \ld{M}(x)
    & =  \bigcup_{\nm\in M}\labelf{\nm}\circ \labelfRg{\nm}\circ
    \invf{\nm}(x)\cap \ld{M}(x)\\
    &= \bigcup_{\nm\in M}\bigl(\labelf{\nm}\circ \labelfRg{\nm}\circ
    \invf{\nm}(x)\setminus \labelf{\nm}\circ (\labelfRg{\nm}\cap \labelfKg{\nm})\circ
    \invf{\nm}(x)\bigr)\cap \ld{M}(x).
  \end{align*}
  Finally, by using the obvious fact that $f(A)\setminus f(A\cap
  B)\subseteq f(A\setminus B)$ for any function $f$, we get
  \[\labelf{H}(x)\cap \ld{M}(x) \subseteq \bigcup_{\nm\in
      M} \labelf{\nm}\circ (\labelfRg{\nm}\setminus \labelfKg{\nm})
      \circ\invf{\nm}(x)\cap \ld{M}(x) \subseteq \llift{M}(x)\]
    hence $\labelf{H}(x)\cap \ld{M}(x)\setminus
  \llift{M}(x)=\ensvide$. This proves that $H$ is disjoint from
  $\Vd{M}$, $\Ad{M}$, $\ld{M}\setminus \llift{M}$ and therefore
  that $M$ has the effective deletion property.
\end{proof}

Yet parallel coherence is not sufficient in the light of parallel
independence, as we now show.

\begin{prop}
  Parallel coherence does not generalize parallel independence.
\end{prop}
\begin{proof}
  In Example~\ref{ex-EDP} is exhibited a set $M=\set{\mm_1,\mm_2}$
  that is shown to be parallel independent. But we see that
  \[\mm_1(R_1\gcap K_1)\gcap\mm_2(L_2) = \mm_1(\,\raisebox{-1.3ex}{\begin{tikzpicture}
    \node[draw,rounded corners] at (0,0)
    {$x_1\sepattr u$};
  \end{tikzpicture}}\,) \gcap \mm_2(\,\raisebox{-1.3ex}{\begin{tikzpicture}
    \node[draw,rounded corners] at (0,0)
    {$x_2\sepattr v$};
  \end{tikzpicture}}\,) = \raisebox{-1.3ex}{\begin{tikzpicture}
    \node[draw,rounded corners] at (0,0)
    {$x\sepattr 0$};
  \end{tikzpicture}} = G\] is not a subgraph of
  $\mm_2(K_2) = \raisebox{-0.7ex}{\begin{tikzpicture}
    \node[draw,rounded corners] at (0,0) {$x$};
  \end{tikzpicture}}$, hence $M$ is not parallel coherent.
\end{proof}

Parallel coherence is therefore too restricted to support a definition
of parallel rewriting in the present framework. The problem here as
above is that deleted attributes can be recovered by the right-hand
side of rules, and that this possibility is not accounted for in the
algebraic definitions, since these do not distinguish between graph
items and attributes.

To summarize, we have established that the following implications
hold, and no other:
\[\text{regularity }\imply \text{ parallel coherence } \imply
  \text{ effective deletion property } \Leftarrow \text{ parallel independence.}\]

We see this as an endorsement of parallel rewriting based on the
effective deletion property (Definition~\ref{def-EDP}), even if it is
the only property that cannot be characterized simply on pairs of
matchings.  This suggests that the effective deletion property would
be worth transposing to an algebraic framework. But there is no
straightforward way of doing this, as can now be shown.

\begin{corollary}
  For any set of unlabeled rules $\R$, any unlabeled graph $G$ and
  any subset $M\subseteq \Matches{\R}{G}$, 
   \[M\text{ is regular\ \ iff\ \ }M\text{ is parallel
   coherent\ \ iff\ \ }M\text{ has the effective deletion property.}\]
\end{corollary}
\begin{proof}
  Assume that $M$ has the effective deletion property, then
  $\Ipgr{G}{M}$ is disjoint from
  $\Vd{M},\Ad{M},\ld{M}\setminus\llift{M}$ hence so is
  $\Gcup_{\nm\in M}\liftm{\nm}(\Rg{\nm})$. For all $\mm\in M$ we have
  $\Vd{\mm}\subseteq \Vd{M}$ and $\Ad{\mm}\subseteq \Ad{M}$, hence
  $\Gcup_{\nm\in M}\liftm{\nm}(\Rg{\nm})$ is disjoint from $\Vd{\mm}$,
  $\Ad{\mm}$, $\ensvide$ and therefore so is $\liftm{\nm}(\Rg{\nm})$
  for every $\nm\in M$. Thus
  $\liftm{\nm}(\nodesRg{\nm})\cap \mm(\nodesLg{\mm})\setminus
  \mm(\nodesKg{\mm}) = \ensvide$ and
  $\liftm{\nm}(\arrowsRg{\nm})\cap \mm(\arrowsLg{\mm})\setminus
  \mm(\arrowsKg{\mm}) = \ensvide$, which is equivalent to
  $\liftm{\nm}(\nodesRg{\nm})\cap \mm(\nodesLg{\mm})\subseteq
  \mm(\nodesKg{\mm})$ and
  $\liftm{\nm}(\arrowsRg{\nm})\cap \mm(\arrowsLg{\mm})\subseteq
  \mm(\arrowsKg{\mm})$. Since these graphs are unlabeled, this
  entails that
  $\liftm{\nm}(\Rg{\nm})\gcap \mm(\Lg{\mm})\subalg \mm(\Kg{\mm})$ for
  all $\mm,\nm\in M$, hence that $M$ is regular by
  Lemma~\ref{lm-regular}. The equivalences follow by
  Lemma~\ref{lm-reg2parcoh} and Theorem~\ref{thm-parcoh2edp}.
\end{proof}
Hence an algebraic approach to parallel graph transformation that
would apply to the category of (unlabeled) graphs could not
distinguish these notions. In this sense parallel coherence is already
the right algebraic translation of the effective deletion property
(and of regularity), even if it is too weak to account for the special
treatment of attributes in the present non algebraic framework.

\section{Related Work and Conclusion}\label{sec-conclusion}

Many notions of attributed graphs exist in the literature. For
instance, in \cite{Plump04,DuvalEPR14} graph items can hold at most
one attribute. This means that concurrent rules could possibly conflict
because of their right-hand sides, if two rules required to attribute
distinct values to the same graph item. Our choice of attaching sets
of attributes to vertices and arrows means that new attributes are
freely included in those sets, and thus avoids conflicting right-hand
sides. Indeed, we see from Definition~\ref{def-EDP} that a
conflict must involve an element of $\Vd{M}$, $\Ad{M}$ or
$\ld{M}$. Hence the right-hand sides of rules never \emph{create}
conflicts, though they may \emph{settle} the conflicts created in the
left-hand sides and are therefore relevant to
parallel independence.

Other notions of attributed graphs that allow unbounded attributes are
possible, for instance the E-graphs from \cite{EhrigEPT06}. But the
fact that in E-graphs a single value can be referenced several times
as attribute of a vertex or arrow means that the number of matchings
of rules may uselessly inflate.

Another approach to parallelism is to accept overlapping, non
independent matchings and ask the user to decide what to do in
particular situations \cite{KniemeyerBHK07}. The present
approach shows that the user can be spared this work not just on
parallel independent matchings, but on the larger class of sets that
satisfy the effective deletion property (or parallel coherence in an
algebraic framework). 

It is also possible to restrict by design all
overlaps to vertices, as is the case in Hyperedge Replacement Systems
\cite{DrewesKH97}, and still be able to specify powerful parallel
transformations \cite{LaneseM05}, though in a non deterministic
way. Note that these are asynchronous models of parallelism, where
determinism amounts to confluence. This property has been widely
studied in term rewriting; it becomes more subtle when acyclic term
graphs are considered \cite{Plump99}, and more elusive when cycles are
allowed \cite{AriolaK96,AriolaB97}. Our model of parallelism is a
synchronous one where deterministic transformations can be designed
without reference to confluence \cite{BdlTE20}, as in cellular automata.

The use of parallel transformations to define sequential independence
in an algebraic approach to graph rewriting (as in
Definition~\ref{def-seqindep}) could be worth investigating.

\bibliographystyle{eptcs}

\end{document}